\documentclass[pra,longbibliography,twocolumn,showpacs,nofootinbib,superscriptaddress,notitlepage]{revtex4-1}
\usepackage{amsmath}
\usepackage{amssymb,bm}
\usepackage{amsthm}
\usepackage{color,dsfont}
\usepackage{graphicx} 
\usepackage[colorlinks=true, hyperindex, breaklinks, linkcolor=blue, urlcolor=blue, citecolor=blue]{hyperref} 
\usepackage[normalem]{ulem}
\usepackage[capitalise]{cleveref}
\usepackage{mathrsfs}
\usepackage{subcaption}
\usepackage{mathtools}

\newcommand{\CPLX}{\mathbb{C}}

\newtheorem{claim}{Claim}

\newcommand{\ket}[1]{|#1\rangle}  
 
\newcommand{\dla}{\ensuremath{\langle\hspace{-.25em}\langle}}
\newcommand{\dra}{\ensuremath{\rangle\hspace{-.25em}\rangle}}

\newcommand{\e}{\epsilon}

\newcommand{\codepar}[1]{\ensuremath{[\![#1]\!]}}

\begin{document}

\title{Error suppression via complementary gauge choices in Reed-Muller codes}

\author{Christopher Chamberland}
\email{c6chambe@uwaterloo.ca}
\affiliation{
    Institute for Quantum Computing and Department of Physics and Astronomy,
    University of Waterloo,
    Waterloo, Ontario, N2L 3G1, Canada
    }
\author{Tomas Jochym-O'Connor}
\email{tjoc@caltech.edu}
\affiliation{
   Walter Burke Institute for Theoretical Physics and Institute for Quantum Information and Matter,
   California Institute of Technology, Pasadena, California 91125, USA
    }

\begin{abstract}
Concatenation of two quantum error correcting codes with complementary sets of transversal gates can provide a means towards universal fault-tolerant computation. We first show that it is generally preferable to choose the inner code with the higher pseudo-threshold in order to achieve lower logical failure rates. We then explore the threshold properties of a wide range of concatenation schemes. Notably, we demonstrate that the concatenation of complementary sets of Reed-Muller codes can increase the code capacity threshold under depolarizing noise when compared to extensions of previously proposed concatenation models. We also analyze the properties of logical errors under circuit level noise, showing that smaller codes perform better for all sampled physical error rates. Our work provides new insights into the performance of universal concatenated quantum codes for both code capacity and circuit level noise. 
\end{abstract}


\maketitle

\section{Introduction}
\label{sec:Intro}

Quantum error correcting codes provide a means to suppress errors in physical quantum systems through the encoding of information into protected subspaces of a larger Hilbert space. Choosing how to encode such information however will greatly affect the rate and behaviour of such error suppressions. Moreover, codes must be carefully chosen in order to prevent the propagation of errors when trying to manipulate the stored information for the purposes of performing logical operations.

In this work, we extend the constructions first presented and analyzed in Refs.~\cite{JL14, CJL16} for concatenated codes that can implement a universal set of fault-tolerant operations. Namely, we expand on the possible combinations of different error correcting codes that can provide fault-tolerant universality, determining their threshold behaviour and logical failure rates given two different error models, code capacity and circuit level depolarizing noise. The original construction used complementary codes, the 7-qubit Steane and 15-qubit Reed-Muller, as they have different sets of transversal gates. These codes also correspond to the smallest codes of the 2D and 3D color code families~\cite{BM07, BM07c}. As such, the simplest generalization would be to consider larger codes from such families and determine their threshold behaviours. 

Alternatively, we show that by concatenating the 15-qubit Reed Muller code with a rotated version of the same code, we can obtain two single-qubit non-Clifford fault-tolerant gates which allow for the implementation of a universal set of gates. This rotated version of the Reed-Muller code can be thought of as fixing a new gauge for the code, where some of the $Z$~stabilizers of the original code are replaced by $X$~stabilizers. Moreover, such a concatenation configuration exhibits a higher threshold than combining the 15-qubit Reed-Muller code with different 2D~color codes in the case of code capacity noise.

The concatenation order for a universal scheme matters for the resulting value of the fault-tolerance threshold. In particular, given two codes with different pseudo-thresholds, it is generally preferable to take the inner code to be the one with the higher pseudo-threshold in order to get preferable logical noise suppression. We begin by arguing from an asymptotic coding perspective why such a behaviour should hold true in Section~\ref{sec:ConcatOrder}, and expand upon this argument with exact analytic and numeric simulations. In Section~\ref{sec:Circuit15Noise} we detail the underlying circuits for the universal model based on the concatenation of complementary Reed-Muller codes, and analyze the logical noise behaviour that results from such a concatenation. In order to optimize the performance under gate noise, we derive a new state preparation algorithm as well as logical gate constructions (see \cref{app:EncodedCircuits,app:MinCNOTforTandHTH,app:MinCNOTDepth}). Namely, the logical gate constructions for the non-Clifford gates are optimal with respect to their number of CNOT~gates as well as circuit depth.

\section{Optimizing the concatenation order for universal concatenated quantum codes}
\label{sec:ConcatOrder}

In \cref{subsec:OptimizeConcat} we begin by presenting an argument for determining the concatenation order of two different error correcting codes that yields the lowest logical failure rate. In \cref{subsec:PauliChanG}, we provide a brief review of the process matrix formalism and show how their calculation can be simplified for physical noise models described by Pauli channels. Using symmetries of error correcting codes and optimization methods developed in Ref.~\cite{CWBL16}, we compute the code capacity depolarizing noise\footnote{Code capacity noise corresponds to memory noise on each qubit with perfect encoding, error correction, and decoding.} thresholds for the codes listed in \cref{tab:StabilizerGeneratorsLists}. In \cref{subsec:UniversalQuantumCodes}, we provide new insights into constructing universal concatenated quantum codes. We use the optimization tools of this section to compute memory thresholds of several universal concatenated quantum codes for a depolarizing noise model. These results demonstrate that the concatenation order of two different error correcting codes can affect the threshold for a biased noise model.

\subsection{Optimizing the concatenation order of error correcting codes}
\label{subsec:OptimizeConcat}

Recall, the pseudo-threshold is the physical error rate at which a particular quantum error correction encoding exhibits a lower logical error rate than the physical error rate. However, an subtly important different quantity is the asymptotic threshold, which is often times referred to as just \textit{the threshold}, which is the physical error rate at which a family of codes begin to show exponential suppression of the logical error rate as a function of the distance of the code. In case of concatenated codes, we can consider multiple rounds of concatenation, and check when the error rate of the second concatenation crosses below that of the first to obtain a bound of asymptotic threshold, as in Ref.~\cite{CJL16}. In this work, we will be studying memory noise, that is when only the data qubits undergo independent noise, which is often referred to a code capacity noise. Additionally, we will study circuit level noise, where every circuit element undergoes independent noise in addition to memory noise, which provides a more realistic modeling of the noise that would be seen in experimental realizations. The calculated thresholds in both of these cases will be referred to appropriately as the code capacity or circuit level threshold, and whether they are pseudo or asymptotic thresholds will be specified.

Consider the codes $C_{1}$ and $C_2$ with code parameters $\codepar{n_{1},k_{1},d_{1}}$ and $\codepar{n_{2},k_{2},d_{2}}$, respectively. Suppose we concatenate the codes $C_{1}$ and $C_{2}$, see Refs.~\cite{JL14, CJL16}. In this work, we explore the freedom in which we can choose the concatenation order of the two codes, that is the specific choice of which code should be the inner and outer code. Here, we present an argument as to why it is generally favourable to choose a code with a higher pseudo-threshold as the inner code.

Consider the case where the code $C_1$ forms the outer code, while $C_2$ forms the inner code, that is the qubits forming the outer code~$C_1$ are themselves encoded in the code~$C_2$. Suppose that $C_2$ has a higher pseudo-threshold than $C_1$, where the pseudo-threshold is the threshold rate at which the logical error rate of the code becomes smaller than the physical error rate. Denote these threshold rates as $\e_1$, and~$\e_2$. We argue that it is favourable that $\e_1 < \e_2$ since there will be a regime of physical error rates $\e_1 < p < \e_2$ that may be otherwise uncorrectable. In this regime, the logical error rate for inner code~$C_2$ will satisfy $p_2(p) < p$. The rate at which $p_2(p)$ is suppressed will be a function of the code used and other important fault-tolerance factors, however if the suppression is high enough it could be that $p_2 < \e_1$. In such a case, since the logical error rate is below the pseudo-threshold of the outer code~$C_1$, noise will be further suppressed and the error rate~$p$ is below the pseudo-threshold of the concatenated code. Of course, this argument relies on the logical error rate of the inner code~$C_2$ being suppressed rapidly below threshold, which is always true as the code size grows for families of topological codes, however we found it also to be true via simulation for small topological codes.

Consider now if we chose the alternative encoding order. Then since $\e_1<p$, the logical error rate would actually increase for the inner code, $p_1 (p) > p$. Moreover, if the logical error rate were to increase sufficiently rapidly, it would be in a regime where $\e_2 < p_1$, thereby no longer being below threshold for the new outer code. As such, there would be no hope for further suppressing errors, resulting in the logical error rate being worse than the physical error rate.

\subsection{Effective noise for Pauli channels}
\label{subsec:PauliChanG}

For the remainder of this section, we will focus on Pauli channels and perform an exact error analysis using the process matrix formalism developed in Refs.~\cite{RDM02,CWBL16}. The process matrix formalism allows one to compute the effective noise at the logical level after performing error correction given a noise channel acting on all of the physical qubits. To be concrete, suppose the physical noise channel is given by $\mathcal{N}$. For a stabilizer code $C$ encoding a single logical qubit into $n$ physical qubits, we can write the encoding map $\mathcal{E}: \mathcal{H}_{2} \to \mathcal{H}_{C}\subset \mathcal{H}_{2^{n}}$ as

\begin{align}
\mathcal{E}(\rho_{in})=B\rho_{in}B^{\dagger},
\label{eq:EncodingMap}
\end{align}
where $B=\ket{\overline{0}} \langle 0| + \ket{\overline{1}} \langle 1|$. A pure input state $\ket{\psi} = \alpha \ket{0} + \beta \ket{1}$ is encoded to $\ket{\overline{\psi}} = \alpha \ket{\overline{0}} + \beta \ket{\overline{1}}$ under the map $\mathcal{E} \in \mathcal{H}_{C}$.

The decoding step can be separated into two parts. Given a measured syndrome $l$ with associated recovery map $R_{l}$, we define $\mathcal{R}_{l}$ as the map that includes the measurement update and recovery $R_{l}$. The map $\mathcal{E}^{\dagger}$ simply decodes the encoded state back to $\mathcal{H}_{2}$. The effective noise channel can then be written as 

\begin{align}
\mathcal{G}(\mathcal{N},R_{l}) = \mathcal{E}^{\dagger} \circ \mathcal{R}_{l} \circ \mathcal{N} \circ \mathcal{E}.
\label{eq:EffectiveChannel}
\end{align}

Using the set of normalized Pauli matrices $\boldsymbol{\sigma} = (I,X,Y,Z)/\sqrt{2}$, it was shown in Ref.~\cite{CWBL16} that the matrix representation of $\mathcal{G}(\mathcal{N},R_{l})$ is given by

\begin{align}
\boldsymbol{\mathcal{G}}_{\sigma,\tau}(\mathcal{N},R_{l}) = \sum_{\sigma \in \boldsymbol{\sigma}} |\mathcal{G}(\mathcal{N},R_{l})(\sigma) \dra \dla \sigma |,
\label{eq:MatrixRepresentationG}
\end{align}
where the map $|.\dra:\CPLX^{d\times d}\to\CPLX^{d^2}$ is defined by setting $|B_{j} \dra = \boldsymbol{e}_{j}$ and the set $\{ \boldsymbol{e}_{j} \}$ corresponds to the canonical unit basis of $\CPLX^{d^{2}}$.

The effective process matrix presented in \cref{eq:EffectiveChannel} is dependent on the recovery map $R_{l}$ that was chosen for the measured syndrome $l$. The full process matrix can be obtained by averaging over all syndromes. 

\begin{table*}
\begin{tabular}{ c|c|c|c}
 5-qubit code & Steane code & 15-qubit Reed-Muller & 17-qubit color code \\ \hline
     XZZXI      &   IIIZZZZ     & ZIZIZIZIZIZIZIZ & ZZZZIIIIIIIIIIIII \\
     IXZZX      &   IZZIIZZ     & IZZIIZZIIZZIIZZ & ZIZIZZIIIIIIIIIII \\
     XIXZZ      &   ZIZIZIZ     & IIIZZZZIIIIZZZZ & IIIIZZIIZZIIIIIII \\
     ZXIXZ      &   IIIXXXX    & IIIIIIIZZZZZZZZ & IIIIIIZZIIZZIIIII \\
	             &   IXXIIXX    & IIZIIIZIIIZIIIZ     & IIIIIIIIZZIIZZIII \\
                     &   XIXIXIX    & IIIIZIZIIIIIZIZ     & IIIIIIIIIIZZIIZZI \\
                     &                   & IIIIIZZIIIIIIZZ     &  IIIIIIIZIIIZIIIZZ \\
                     &                   & IIIIIIIIIZZIIZZ     &  IIZZIZZIIZZIIZZII \\
                     &                   & IIIIIIIIIIIZZZZ     &  XXXXIIIIIIIIIIIII \\
                     &                   & IIIIIIIIZIZIZIZ     &  XIXIXXIIIIIIIIIII \\
                     &                   & XIXIXIXIXIXIXIX  &  IIIIXXIIXXIIIIIII \\
                     &                   & IXXIIXXIIXXIIXX  &  IIIIIIXXIIXXIIIII \\
                     &                   & IIIXXXXIIIIXXXX  &  IIIIIIIIXXIIXXIII \\
                     &                   & IIIIIIIXXXXXXXX  &  IIIIIIIIIIXXIIXXI \\
                    &                    &                             &  IIIIIIIXIIIXIIIXX  \\
                    &                    &                             &  IIXXIXXIIXXIIXXII  \\\hline
$\langle HS,X,Z\rangle$ & $\langle H,S,CNOT\rangle$ & $\langle 
T,X,Z,CNOT\rangle$ & $\langle H,S,CNOT\rangle$
\end{tabular}
\caption{Stabilizer generators (top) and the generating set of transversal logical operations (bottom) for the 5-qubit 
code~\cite{LMPZ96}, Steane's 7-qubit code~\cite{Steane96b}, the 15-qubit Reed-Muller code~\cite{ADP14}, and the 17-qubit color code~\cite{Bombin15} (see Fig.~\ref{fig:2DCC17}), where $H$ and $S$ are the Hadamard and phase gates respectively, $T = \mathrm{diag}(1,e^{i\pi/4})$, and $\langle.\rangle$ denotes the group generated by the argument.}
\label{tab:StabilizerGeneratorsLists}
\end{table*}

Pauli channels have the form:
\begin{align}
\mathcal{N}(\rho) = (1-p_{x}-p_{y}-p_{z})\rho + p_{x}X\rho X + p_{y}Y\rho Y + p_{z}Z \rho Z.
\label{eq:PauliChan}
\end{align}
By choosing the Pauli basis, the matrix representation of the channel in \cref{eq:PauliChan} is diagonal and can be expressed as 
\begin{align}
\boldsymbol{\mathcal{N}} &= \left( \begin{array}{cccc}
                                          1 & 0 & 0 & 0  \\
                                          0 & 1-2p_{y}-2p_{z} & 0 & 0\\  
                                          0 & 0 & 1-2p_{x}-2p_{z} & 0\\
                                          0 & 0 & 0 & 1-2p_{x}-2p_{y}\\                                                                   
                                          \end{array} \right) \nonumber \\
&= \mathrm{diag}(1,x,y,z).
\end{align}
The derivation of the matrix representation of \cref{eq:MatrixRepresentationG} is significantly simplified for Pauli channels. We first define the function
\begin{align}
f_{j,\sigma}(R_{l}) = \eta(S_{j},R_{l})\eta(R_{l},\overline{\sigma}),
\label{eq:FunctionCommute}
\end{align}
where given the Pauli operators $P_{1}$ and $P_{2}$, $\eta(P_{1},P_{2}) = \pm1$ for $P_{1}P_{2}=\pm P_{2}P_{1}$. The operators $S_{j}$ belong to the stabilizer group and $\overline{\sigma}$ is an element in the set of the codes logical operators. Lastly, we define $w_{\sigma}(P)$ to be the $\sigma$ weight of the Pauli operators $P$. For instance, $w_{Z}(Z_1X_2Y_3Z_4Z_5) = 3$.

With the above definitions, and for an $n$-qubit channel of the form $\mathcal{N} = (\mathcal{N}^{(1)})^{\otimes n}$ where $\mathcal{N}^{(1)}$ is given by \cref{eq:PauliChan}, it was shown in \cite{RDM02} that the components of the process matrix are given by

\begin{align}
\boldsymbol{\mathcal{G}}_{\sigma,\tau}(R_{l}) = \delta_{\sigma,\tau} \frac{1}{|S|}\sum_{j} f_{j,\sigma}(R_{l})x^{w_{X}(S_{j}\overline{\sigma})}y^{w_{Y}(S_{j}\overline{\sigma})}z^{w_{Z}(S_{j}\overline{\sigma})}.
\label{eq:ProcessMatComponentPauli}
\end{align}

It is also straightforward to compute the process matrix for concatenated codes. In Ref.~\cite{RDM02}, it was shown that by computing the effective noise channel at the first concatenation level $\mathcal{G}^{(1)}(\mathcal{N})$ , the channel at the second level could be computed by replacing the physical noise $\mathcal{N}$ by $\mathcal{G}^{(1)}$ in \cref{eq:EffectiveChannel}. This process can be repeated recursively to obtain the effective noise channel at any desired concatenation level. 

For an $\codepar{n,k,d}$ stabilizer code, the number of distinct syndromes is given by $2^{n-k}$. If a code encodes a single logical qubit, one would be required to compute $2^{n-1}$ process matrices in \cref{eq:ProcessMatComponentPauli} for each recovery map $R_{l}$. However, as was shown in Ref.~\cite{CWBL16}, not all process matrices are distinct. Given a generic noise model which is independently and identically distributed (i.i.d.) on all the physical qubits, there can be many process matrices for distinct syndromes which are identical. For a generic Pauli channel, the number of distinct process matrices is even smaller. It is also important to point out that the number of distinct process matrices depends on the particular decoder that is being used, that is, the choice of recovery maps $R_{l}$ for a given syndrome measurement $l$. One particular type of decoder which will be used below is the symmetric decoder. By symmetric decoder, we are considering a decoder that associates the measured syndrome with the error that acts on the fewest number of qubits and is consistent with the syndrome. As an example, the standard decoder for CSS codes consists of correcting $X$ and $Z$ errors independently. Decoding in this way would yield lower logical $X$ and $Z$ failure rates compared to logical $Y$ failure rates. In contrast, applying the symmetric decoder to CSS codes would correct symmetrically between $X$, $Y$ and $Z$ errors so that the logical failure rates of all Pauli's would be symmetric. This is achieved by associating syndromes to some (but not all) errors of the form $X_{i}Z_{j}$ to errors of the form $X_{s}Y_{t}$ or $Z_{u}Y_{v}$.  Note that for the concatenated codes studied in this section, the symmetric decoder is applied independently at each concatenation level and is thus is suboptimal when compared with the min-weight decoder that accounts for syndrome information from both concatenation layers, see Sec.~\ref{subsec:NoiseAndDecoding}.

\begin{figure}
\centering
\includegraphics[width=0.30\textwidth]{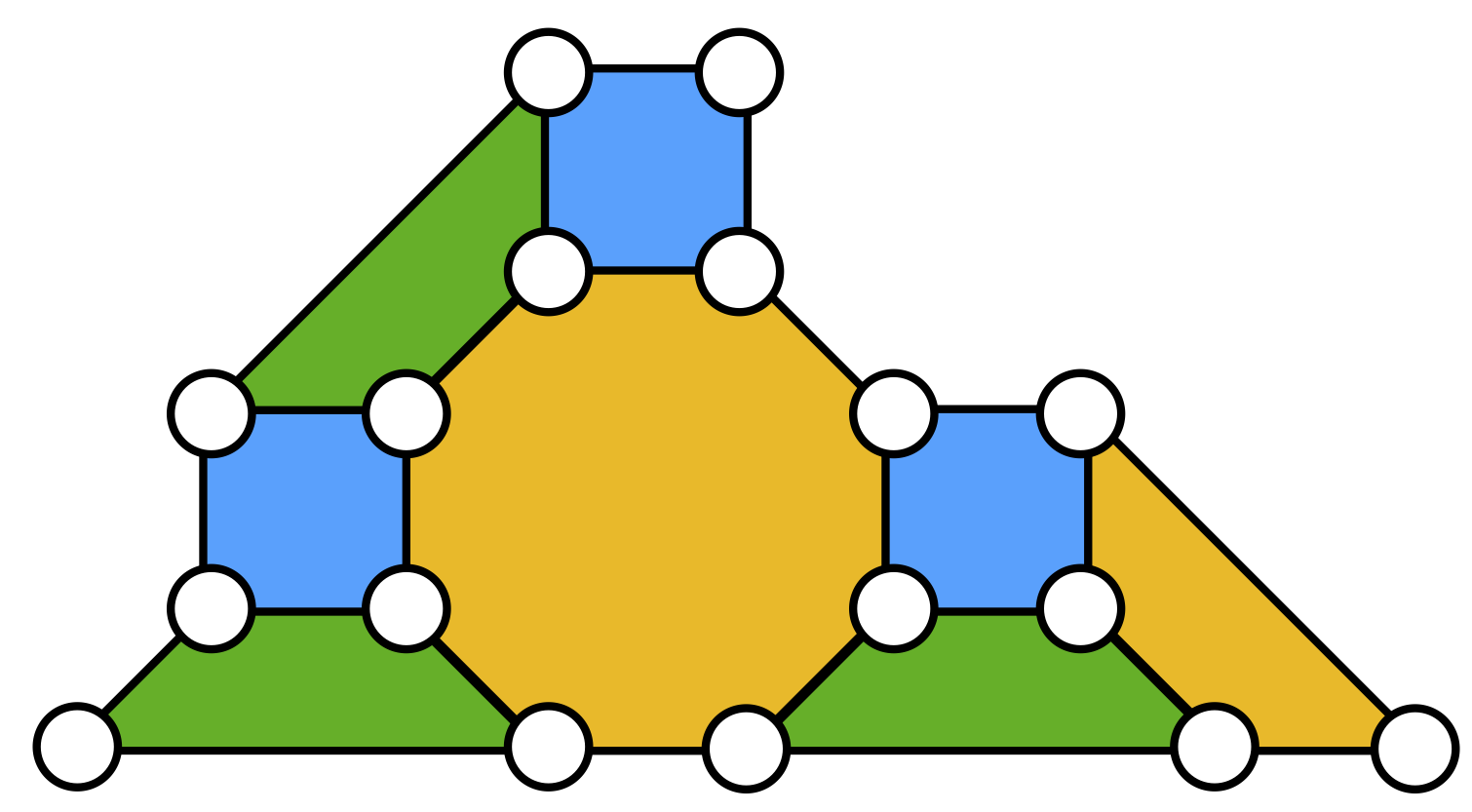}
\caption{Graphical representation of the 17-qubit 2D color code. The $X$~and $Z$~stabilizers of the code are symmetric, given by many weight-4 and a single weight-8 plaquette. The distance of the code is 5 as it scales with the length of each of the sides of the triangular lattice.}
\label{fig:2DCC17}
\end{figure}

In \cref{tab:DistinctProcessMatTab}, we computed the number of distinct process matrices for the codes of \cref{tab:StabilizerGeneratorsLists}. To find the total number of distinct process matrices, we computed $\boldsymbol{\mathcal{G}}_{\sigma,\tau}(R_{l})$ for each syndrome $R_{l}$ using a generic i.i.d. Pauli channel and a symmetric decoder. Comparing the full list of the obtained process matrices allowed us to determine which ones were distinct. Comparing the total number of possible syndrome measurements with the number of distinct process matrices, it can be seen that substantial savings can occur.

\begin{table}
\begin{tabular}{|c|c|c|c|c|}
 \hline
 &\codepar{5,1,3} & \codepar{7,1,3} & \codepar{15,1,3} & \codepar{17,1,5} \\ \hline
Number of syndromes & 16 & 64 & 16384 & 65536  \\ \hline
Distinct process matrices & 4 & 7 & 23 & 106 \\  \hline
\end{tabular}
\caption{Table for the number of distinct process matrices for the codes in \cref{tab:StabilizerGeneratorsLists}. We used a symmetric decoder and a generic Pauli channel. Comparing with the total number of distinct syndromes, it can be seen that substantial savings can occur.}
\label{tab:DistinctProcessMatTab}
\end{table}

\begin{table}
\begin{tabular}{|c|c|c|c|c|}
 \hline
 &\codepar{5,1,3} & \codepar{7,1,3} & \codepar{15,1,3} & \codepar{17,1,5} \\ \hline
Threshold & $p=0.1835$ & $p=0.1291$ & $p=0.0254$ & $p=0.1608$  \\ \hline
\end{tabular}
\caption{Exact depolarizing noise thresholds for the codes given in table \cref{tab:StabilizerGeneratorsLists}. We note that codes with higher thresholds also have higher pseudo-thresholds. The depolarizing noise channel is given by \cref{eq:DepolarisingNoiseEQ}. The 15-qubit code has the lowest depolarizing noise threshold value, due to its high number of qubits and low distance.}
\label{tab:DepThresholdTable}
\end{table}

Using the above tools, we now compute code capacity thresholds of the codes in \cref{tab:StabilizerGeneratorsLists} for a depolarizing noise model given by
\begin{align}
\mathcal{N}(\rho) = (1-\frac{3}{4}p)\rho +\frac{p}{4}(X \rho X + Y \rho Y + Z \rho Z).
\label{eq:DepolarisingNoiseEQ}
\end{align}
In the infinite concatenation limit, we demand that the effective noise channel converges to the identity channel. This ensures that the initial state can be perfectly recovered in the presence of noise. Therefore, a threshold will correspond to the value $p_{th}$ such that for $p<p_{th}$, the process matrix converges to the identity matrix in the infinite concatenation limit.

In computing the effective noise channel at each concatenation level, we used an optimized hard\footnote{In concatenated coding, hard decoding refers to decoding errors layer by layer, where the correction at one concatenation layer is independent of the corrections made at another layer.} decoding algorithm developed in Ref.~\cite{CWBL16} to achieve lower logical failure rates for the studied error correcting codes. In \cref{tab:DepThresholdTable}, we computed exact threshold values for the codes listed in \cref{tab:StabilizerGeneratorsLists}. It can be seen that the 15-qubit code has a lower threshold value by roughly an order of magnitude than the other codes that were considered. The reason is that the 15-qubit code encodes one logical qubit using 15 physical qubits, and the effective distance is only 3 (compared to the 17-qubit color code which is 5). The large number of physical qubits means there are more opportunities for errors to be introduced which cannot be corrected due to the code's low distance.

In \cref{subsec:UniversalQuantumCodes}, we will use the tools developed in this section and in \cref{subsec:OptimizeConcat} in order to find an optimized concatenation scheme for the codes in \cref{tab:StabilizerGeneratorsLists} that lead to universal concatenated quantum codes. Despite the low threshold of the 15-qubit code, we will show that we can use its asymmetric distance in $X$ and $Z$ to our advantage in order to obtain a code with a relatively high memory threshold compared to other concatenation schemes.

\subsection{Concatenation order analysis of universal concatenated quantum codes}
\label{subsec:UniversalQuantumCodes}

\begin{table*}
\begin{tabular}{ |c|c|c|c|c|c|c|c|c|c|}
\hline
& [5,15] & [15,5] & [7,15] & [15,7] & [15,15H] & [17,15] & [15,17] & [7,17] & [17,7] \\ \hline
Depolarizing threshold & $p=0.1146$ & $p=0.1393$ & $p=0.04768$ & $p=0.06886$ & $\boldsymbol{p=0.1065}$ & $p=0.05993$ & $p=0.0997$ & $p=0.1523$ & $p=0.1425$ \\ \hline
\end{tabular}
\caption{Depolarizing threshold results for a set of universal concatenated quantum codes (apart from the last two columns). The depolarizing noise channel is given by \cref{eq:DepolarisingNoiseEQ}. We defined the notation $[n_{1},n_{2}]$ to indicate that the code $C_{1} = \codepar{n_{1},1,d_{1}}$ forms the outer code and  $C_{2} = \codepar{n_{2},1,d_{2}}$ forms the inner code. Note that 15H indicates the Hadamard transform of the 15-qubit Reed-Muller code. All other codes are chosen from \cref{tab:StabilizerGeneratorsLists}. For universal concatenated codes where the CNOT gate is transversal in both codes, it can be observed that [15,15H] has the highest memory threshold.}
\label{tab:DepolarizingThresholdResults}
\end{table*}

The idea behind universal concatenated quantum codes is to use a modified notion of transversality which demands only that single-qubit errors remain correctable \cite{JL14}. By concatenating two different error correcting codes, a gate that cannot be implemented transversally in one code can be implemented transversally in the other and vice-versa. Hence all gates in the universal gate set can be implemented fault-tolerantly.

In the remainder of this section, we concatenate pairs of stabilizer codes chosen from \cref{tab:StabilizerGeneratorsLists} that form universal concatenated quantum codes. The 5-qubit code is the cyclic perfect code, the smallest code that corrects for any single qubit error. The 7-qubit Steane code is the smallest 2D color code, while the 17-qubit code is the smallest 2D color code with distance~$5$, see Fig.~\ref{fig:2DCC17} for a graphical representation of the code where all $X$ and $Z$~stabilizers are constructed by considering plaquettes of the code. Note that all of the plaquettes intersect at two qubits, therefore the stabilizers clearly commute. The 2D color code family has the property that all of the Clifford gates can be implemented transversally. The 15-qubit Reed-Muller code is the smallest code of the 3D color code family. In particular, by fixing a particular gauge in the 3D~color code, that is fixing all of the weight-4 $Z$~stabilizers in Table~\ref{tab:StabilizerGeneratorsLists}, the code contains a transversal $T$~gate. Fixing the gauge can be achieved for higher distance 3D~color codes as well by fixing the basis of the set of faces at the intersections of different tetrahedra in the 3D~gauge color code.

We note an interesting feature about the 15-qubit Reed-Muller code. By taking its Hadamard transform, we obtain a new code (replace all $Z$ stabilizers with $X$ stabilizers and vice-versa) with a generating set of transversal logical operations given by $\langle HTH, X,Z, \text{CNOT} \rangle$, see Appendix~\ref{app:TransversalHTH} for details\footnote{Note, that by exchanging the $X$ and $Z$~stabilizers, this is equivalent to changing the gauge stabilizers from $Z$~to~$X$, recall they were given by the weight-4 stabilizers. A similar result could be obtained by changing the gauge stabilizers to $Y$~stabilizers.}. By concatenating the 15-qubit Reed-Muller code with its Hadamard transform, we also obtain a universal concatenated quantum code since the generating set $\langle T, HTH, X, Z, \text{CNOT} \rangle$ is universal (see \cref{app:CliffordGenerating}). For the remainder of this work, we shall refer to the gate~$T_x = HTH$, as it corresponds to a $T$~gate rotation in the $x$-axis.

In \cref{tab:DepolarizingThresholdResults}, we computed exact memory thresholds using the depolarizing noise model of \cref{eq:DepolarisingNoiseEQ} for a set of universal concatenated quantum codes (apart from the last two columns where thresholds were computed for the [7,17] and [17,7] code). Comparing the [7,17] code with the [17,7] code, taking the 17-qubit color code as the inner code yields a higher threshold value. This is in agreement with \cref{subsec:OptimizeConcat} since as can be seen from \cref{tab:DepThresholdTable}, the 17-qubit code has a higher threshold (and pseudo-threshold) compared to the Steane code. On the other hand, recall that the 15-qubit Reed-Muller code has the lowest threshold value compared to all codes considered in \cref{tab:StabilizerGeneratorsLists}. Therefore, from \cref{subsec:OptimizeConcat}, we expect that setting the 15-qubit code as the outer code when concatenated with other codes of \cref{tab:StabilizerGeneratorsLists} would yield higher threshold values. The threshold values computed for the codes in \cref{tab:DepolarizingThresholdResults} are in agreement with our expectations.


\begin{table}
\begin{tabular}{ |c|c|c|}
\hline
& $p_{x} = 0.001$, $p_{y} = 0.001$ & $p_{z} = 0.001$, $p_{y} = 0.001$\\ \hline
Threshold & $p_{z} = 0.1199$ & $p_{x} = 0.0437$ \\ \hline
\end{tabular}
\caption{Biased noise thresholds for the Pauli channel given in \cref{eq:PauliChan} using the [15,15H] code. The inner code is the Hadamard transform of the 15-qubit Reed-Muller code which has a $d_{Z}=7$ distance. Therefore, the threshold improves significantly when the noise is biased towards $Z$ errors. If the [15H,15] code was used, the threshold values would be reversed.}
\label{tab:BiasedNoiseTab}
\end{table}

As was first observed in Ref.~\cite{NSZ16}, the 15-qubit Reed-Muller code can be concatenated with the 5-qubit code to form a universal concatenated quantum code (since the transversal $HS$ gate of the 5-qubit code can be combined with the transversal  $T$ and CNOT gates of the 15-qubit code to form a universal gate set). These codes have the highest memory thresholds for the universal concatenated codes studied in \cref{tab:DepolarizingThresholdResults}. However, since the 5-qubit code cannot implement the CNOT gate transversally, in a full circuit level noise analysis, the CNOT gate would not benefit from the dual protection of both codes. The dual protection of the CNOT gate was in large part responsible for the high thresholds observed in Ref.~\cite{CJL16}. Hence, even though the 15-qubit Reed-Muller code concatenated with the 5-qubit code has the highest code capacity threshold, we would not expect these codes to be competitive in a fault-tolerant simulation compared to universal concatenated quantum codes which can implement the CNOT gate transversally in both codes $C_{1}$ and $C_{2}$.

The next universal concatenated quantum code with the highest memory threshold in \cref{tab:DepolarizingThresholdResults} is the code obtained by concatenating the 15-qubit Reed-Muller code (chosen as the outer code) with its Hadamard transform (which we denote [15,15H]). Note that neither of the two codes are 2D~color codes and the universal gate set is not composed of the standard Clifford + $T$ gate set. We also note that the code's threshold is more than twice the threshold for [7,15] which is the code that was studied in Ref.~\cite{CJL16}. In the next section we will perform a full circuit level noise analysis of the [15,15H] code and compare its logical failure rates with the [7,15] code. 

To understand the high threshold value of the [15,15H] code, we first note that the [15,15H] code will perform very well against $Z$ errors at the inner level due to the high $d_{Z}=7$ distance. At the outer level, there will be many "left-over" $X$ errors due to the inner codes low $d_{X}=3$ distance. However, this effect will be compensated by the outer codes large $d_{X}=7$ distance which will be able to correct at least three blocks containing logical $X$ errors. For the [15H,15] code (where the Hadamard transform of the 15-qubit Reed-Muller code is used as the outer code), the analysis for $X$ and $Z$ errors is reversed. Since for depolarizing noise models both $X$ and $Z$ errors occur with equal probability, it is not surprising that we find that both the [15,15H] and [15H,15] codes have the same depolarizing noise threshold.

Since 15H is a $\codepar{15,1,(d_{X}=3,d_{Z}=7)}$ code, the order in which it is concatenated with the 15-qubit code matters if the noise is biased. If the noise is biased in such a way that $Z$ errors occur with higher probability than $X$ and $Y$ errors, choosing 15H as the inner code would yield a higher threshold given that 15H has higher $d_{Z}$ distance (see \cref{tab:BiasedNoiseTab}), which is in agreement with the analysis of \cref{subsec:OptimizeConcat}.

\section{Circuit level noise analysis of the [15,15H] code}
\label{sec:Circuit15Noise}

In \cref{sec:ConcatOrder} we saw that the [15,15H] code had the highest memory threshold value of all other universal concatenated quantum codes where the CNOT gate was transversal in both codes. We used the notation [15,15H] to indicate that the Hadamard transform of the 15-qubit Reed-Muller code was taken as in the inner code. 

In this section we provide a full circuit level noise analysis of the [15,15H] code using the methods of Refs.~\cite{PR12,CJL16,CJL16b}. We calculate the logical failure rates of the $T$, $T_x$ and CNOT gates at the first and second concatenation level. We provide a comparison with the [7,15] code where threshold results were calculated in Ref.~\cite{CJL16}. 

\subsection{Fault-tolerant error correction}
\label{subsec:StatePrep}

\begin{figure}
\centering
\begin{subfigure}{0.35\textwidth}
\includegraphics[width=\textwidth]{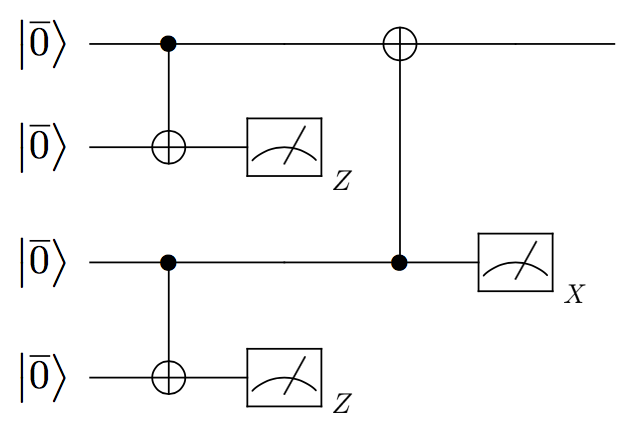}
\caption{}
\label{fig:nOprep}
\end{subfigure}
\begin{subfigure}{0.35\textwidth}
\includegraphics[width=\textwidth]{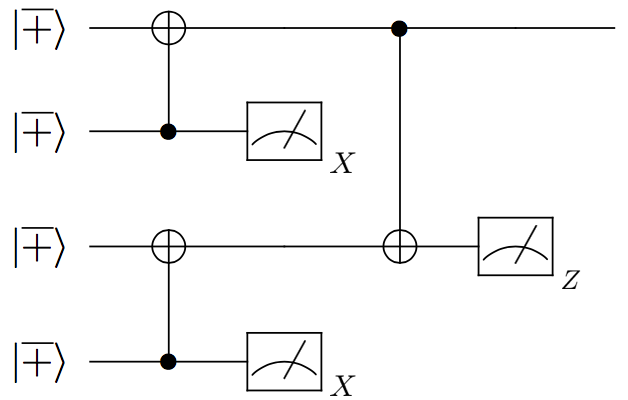}
\caption{}
\label{fig:nPlusPrep}
\end{subfigure}
\caption{Circuits for verifying the encoded $\ket{\overline{0}}$ and $\ket{\overline{+}}$ states as part of Steane error correction. All gates are applied transversally. The extra encoded $\ket{\overline{0}}$ and $\ket{\overline{+}}$ states are verifier states and are used for detecting errors that can occur in the encoding circuits.  If a non-trivial syndrome or the $-1$ eigenvalue of a logical $Z$ or $X$ operator is measured, the states are rejected and the process starts anew with fresh ancilla states.}
\label{fig:nPrep}
\end{figure}

We use Steane's method as part of our fault-tolerant error correction protocol, as the 15-qubit Reed-Muller code is a CSS code. In Steane error correction, $X$ and $Z$ errors are measured separately using encoded $\ket{\overline{0}}$ and $\ket{\overline{+}}$ ancilla states (see \cref{fig:nPrep}). We used Steane's Latin rectangle method \cite{Steane02} in order to compute the encoded $\ket{\overline{0}}$ circuit and applied the optimization methods of \cite{PR12,CJL16b} allowing one to remove some of the CNOT gates. In \cref{app:EncodedCircuits}, we provide a general method for obtaining encoded $\ket{\overline{+}}$ circuits for CSS codes that are not self-dual. The encoding circuits for the $\ket{\overline{0}}$ and $\ket{\overline{+}}$ ancilla states of the 15-qubit Reed-Muller code were obtained in Ref.~\cite{CJL16b} while the circuits for the Hadamard transform of the 15-qubit Reed-Muller code were obtained using the methods of \cref{app:EncodedCircuits}.

Since the Hadamard transform of the 15-qubit Reed-Muller code forms the inner code of [15,15H], the encoding circuits for the full [15,15H] code are obtained by replacing the physical $\ket{0}$ and $\ket{+}$ states in the encoding circuit of the 15-qubit code with the encoded $\ket{\overline{0}}$ and $\ket{\overline{+}}$ states of the Hadamard transform of the 15-qubit code. All of the remaining CNOT gate locations in the encoding circuits of the 15-qubit code are performed transversally.

\subsection{Circuits implementing $T$ and $T_{x}$ gates}
\label{subsec:TandHTHcircuits}

The 15-qubit code can implement all gates in the set generated by $\langle T, CNOT \rangle$ transversally.  However, the $T_x$ gate cannot be implemented transversally since, as was shown in \cref{subsec:UniversalQuantumCodes}, $\langle T, T_x, \mathrm{CNOT} \rangle$ forms a universal gate set. In \cref{app:MinCNOTforTandHTH} we provide a circuit which implements the logical $T_x$ gate (illustrated in \cref{fig:HTHcircuit225}) on the 15-qubit code. Similarly, we provide a circuit which implements the logical $T$ gate on the 15-qubit Hadamard transformed code (illustrated in \cref{fig:Tcircuit225}). Note that in both cases, the circuits have the form $U_{CNOT}T_{x,j}U^{\dagger}_{CNOT}$ for the logical $T_x$ gate and $U_{CNOT}T_{k}U^{\dagger}_{CNOT}$ for the logical $T$ gate where $j,k \in \{1,2, \cdots , 15 \}$. The operator $U_{CNOT}$ is given by a product of CNOT gates acting on the 15-qubit code (or the 15-qubit Hadamard transform code depending on the particular logical gate we are performing). The role of the CNOT gates is to transform part of the stabilizer generators to operators that have no support on the qubit acted on by $T_x$ or $T$. Applying $U^{\dagger}_{CNOT}$ will then guarantee that the resulting operators will be mapped back to operators in the stabilizer group. 

In \cref{app:MinCNOTforTandHTH,app:MinCNOTDepth} we also provide a proof that the operators $U_{CNOT}$ in \cref{fig:TandHTHcircuits} are implemented using the minimum number of gates as well as circuit depth, when using only CNOT gates. This is especially relevant when performing a circuit level noise analysis since the CNOT gates can propagate errors badly within a code block. Therefore, having fewer CNOT gates in the circuit of $U_{CNOT}$ reduces the probability for logical faults to occur. Additionally, minimizing the depth will reduce the total number of memory noise locations, reducing overall error rates.

\subsection{Noise model and decoding algorithm}
\label{subsec:NoiseAndDecoding}

In our full circuit level noise analysis, we considered the depolarizing noise model of \cref{eq:DepolarisingNoiseEQ}. Applying the depolarizing noise model to each location in a quantum circuit, gates, state-preparation, measurements and qubits waiting in memory can fail according to the following noise model:

\begin{enumerate}
   \item A noisy CNOT gate is modelled as applying a CNOT gate followed by, with probability $\frac{15p}{16}$, a two-qubit Pauli error drawn uniformly and independently from $\{I,X,Y,Z\}^{\otimes 2}\setminus \{I\otimes I\}$. 
   \item A noisy preparation of the $\ket{0}$ state is modelled as the ideal preparation of the $\ket{0}$ state with probability $1-\frac{p}{2}$ and $\ket{1}=X\ket{0}$ with probability $\frac{p}{2}$. Similarly, the noisy preparation of the $\ket{+}$ state is modelled as the ideal preparation of the $\ket{+}$ state with probability $1-\frac{p}{2}$ and $\ket{-}=Z\ket{+}$ with probability $\frac{p}{2}$.
   \item A noisy measurement in the $Z$-basis is modelled by applying a Pauli $X$ error with probability $\frac{p}{2}$ followed by an ideal measurement in the $Z$-basis. Similarly, a noisy measurement in the $X$-basis is modelled by applying a Pauli $Z$ error with probability $\frac{p}{2}$ followed by an ideal measurement in the $X$-basis.
   \item A single-qubit gate error or storage error is modelled by applying the ideal gate with probability $1-\frac{3p}{4}$. With probability $\frac{3p}{4}$, the ideal gate is implemented followed by a Pauli error chosen uniformly from the set  $\{ X,Y,Z \}$. 
\end{enumerate}

After measuring the error syndrome, a decoding procedure is implemented in order to correct the errors that could occur. We applied a minimum weight decoder to decode the [15,15H] code. Before proceeding with the description of our decoding scheme, a few definitions are required. 

We define $Q_{l}$ to be the set of all errors corresponding to the measured syndrome value $l$. As an example, for Steane's 7-qubit code and the syndrome $l = 001000$, using the stabilizers of \cref{tab:StabilizerGeneratorsLists} we find that $Q_{l} = \{ X_{1}, X_{2}X_{3}, X_{6}X_{7} ,\dots \}$.

Next, given the stabilizer group $\mathcal{S}$, we define $N(\mathcal{S}) = {\{ P| PQP^{\dagger}}\in \mathcal{S} \> \forall  \> Q \in \mathcal{P}_{n} \}$ to be the normalizer of the stabilizer group ($\mathcal{P}_{n}$ is the $n$-qubit Pauli group). Equivalently, $N(\mathcal{S})$ is the set of all elements that commute with elements of $S$. Consequently, the set $N(\mathcal{S}) \setminus \mathcal{S}$ is the set of all non-trivial logical operators for a code described by the stabilizer group $\mathcal{S}$. 

Given the above definitions we now describe our decoding algorithm:

\begin{enumerate}
   \item Each qubit in the 15-qubit Reed-Muller code is encoded in the 15H code. Correct each of the 15-qubit blocks using the Hadamard transform of the 15-qubit stabilizer generators shown in \cref{tab:StabilizerGeneratorsLists}. Store the total weight of the correction for each block.
   \item  Measure the new syndrome $l$ at the outer level using the stabilizer generators of the 15-qubit Reed-Muller code in \cref{tab:StabilizerGeneratorsLists}. 
   \item For each element in the set $Q_{l}$, correct the corresponding blocks with the appropriate logical operators. Repeat for all logical operators in the set $N(\mathcal{S}) \setminus \mathcal{S}$ of the same type. For every pair of operators in $Q_{l}$ and $N(\mathcal{S}) \setminus \mathcal{S}$, store the total weight correction taking into account corrections from the previous level. 
   \item The final correction will consist of the operators in $Q_{l}$ and $N(\mathcal{S}) \setminus \mathcal{S}$ that gives the overall lowest weight correction. 
\end{enumerate}

As an example, consider the weight 8 error $X_{1}X_{2}$ on the first 4 blocks of the [15,15H] code. Correcting each block according to the stabilizers of the $\codepar{15,1,(3,7)}$ code results in an $X_{3}$ correction on each code block. Hence the total weight correction is 4 (a weight one correction on blocks 1-4). Note that the operator $X_{1}X_{2}X_{3}$ is a logical operator for the $\codepar{15,1,(3,7)}$ code. Using the stabilizers of the $\codepar{15,1,(7,3)}$ code, the error $X_{1}X_{2}X_{3}X_{4}$ triggers the syndrome $l=0100000001$ from the 10 $Z$ stabilizers. For this particular syndrome, $Q_{l}$ is given by the set $Q_{l} = \{ X_{5}X_{6}X_{7}, X_{1}X_{2}X_{3}X_{4}, \dots \}$. Going through all operators in $Q_{l}$ and all logical $X$ operators in the set $N(\mathcal{S}) \setminus \mathcal{S}$, we find that the overall lowest weight correction is to apply $X_{1}X_{2}X_{3}X_{4}$ on the first 4 blocks with the logical $\overline{X}=X_{1}X_{2}X_{3}$. This correction will thus remove the initial weight 8 error. Since $X_{3}$ was applied at the previous level, the total weight correction is $4(3-1) = 8$. If one had chosen $X_{5}X_{6}X_{7}$, the total weight correction would have been $3(3) + 4 = 13$. The factor of 4 comes from the weight one corrections of the previous level. Applying the latter correction would lead to a logical fault.

Note that the minimum weight decoder in this section differs significantly from the symmetric decoder considered in \cref{subsec:PauliChanG}. Instead of correcting independently at each level, we use information from both concatenation levels to apply an \textit{overall} lowest weight correction.

\subsection{Comparison of the logical failure rates of the $T_{x}$ gate with the $H$ gate of the [7,15] code}
\label{subsec:LogicalFailureRates}

\begin{figure}
\centering
\includegraphics[width=0.5\textwidth]{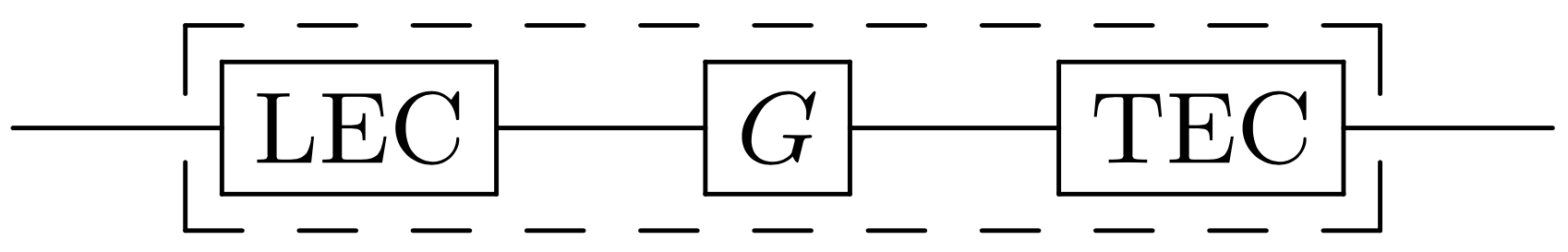}
\caption{Illustration of an extended rectangle (exRec) implementing a logical gate $G$ with its leading (LEC) and trailing (TEC) error correction circuits. The EC circuits perform Steane error correction.}
\label{fig:ECcircuit}
\end{figure}

In a quantum algorithm, one needs to perform fault-tolerant error correction in between the application of logical gates in order to ensure that errors do not spread in an uncontrollable fashion. This ensures that errors of weight $t = \lfloor (d-1)/2 \rfloor$ remain correctable. As was shown in \cite{AGP06}, even though logical gates can share multiple error correction (EC) circuits (such as Steane EC circuits), the threshold of a code is lower bounded by the gate with the highest threshold of all gates in a universal gate set. The gates must be combined with a leading and trailing error correction circuit (which are known as exRec's, an abbreviation that stands for extended rectangles). An illustration of an exRec is shown in \cref{fig:ECcircuit}. Often the exRec with the lowest threshold value will correspond to the CNOT exRec due to the large number of locations (since a CNOT exRec contains four EC circuits compared to two for single-qubit gates). However, as was observed in \cite{CJL16}, for universal concatenated quantum codes where the CNOT gate is transversal in both code's, the CNOT exRec does not necessarily provide a lower bound on the code's threshold. 

In computing logical failure rates, we first define the notion of a malignant event \cite{PR12,CJL16,CJL16b}. We set $\ket{\psi_{1}}$ to be a single or two-qubit logical state obtained by applying ideal decoders immediately after the LEC circuit and $\ket{\psi_{2}}$ to be the logical state obtained by applying ideal decoders immediately after the TEC. We define the event $\mathrm{mal}_{E}$ as $\ket{\psi_{2}} = EU\ket{\psi_{1}}$ where $E$ is a single or two-qubit error and $U$ is the desired logical gate. 

In performing our simulations, at each location of a given exRec $G$, we inserted errors given by the noise model described in \cref{subsec:NoiseAndDecoding}. Once all the error locations were fixed, we propagated the errors through the exRec and ideally decoded the output to determine whether a logical-fault occurred. We determined if the event $\mathrm{mal}_{E}$ occurred for all logical Pauli errors. We repeated our procedure $N$ times in order to compute the probabilities $\mathrm{Pr}(\mathrm{mal}_{E}) = a_{E}/N$. Here $a_{E}$ is the number of times the event $\mathrm{mal}_{E}$ occurred during the $N$ simulations. 
As was shown in Ref.~\cite{CJL16}, for an exRec simulating the gate $G$ under the depolarizing noise channel describe in \cref{subsec:NoiseAndDecoding}, we can upper bound $\mathrm{Pr}(\mathrm{mal}_{E})$ by

\begin{align}
\mathrm{Pr}[\mathrm{mal}^{(1)}_{E}|G,p] \le \sum^{L_{G}}_{k=\lceil \frac{d^{*}}{2} \rceil}c(k)p^{k} \vcentcolon = \Gamma^{(1)}_{G}.
\label{eq:UpperBoundProb}
\end{align}
The coefficients $c(k)$ parametrize the possible weight-$k$ errors that lead to a logical fault, $L_{G}$ is the total number of locations in the circuit $G$ and $d^{*}$ characterizes the minimal distance of a given logical gate. 

When going to the second concatenation level, each level-1 exRec can be treated as a physical location with the effective noise rate given by the polynomials $\Gamma^{(1)}_{G}$. More details can be found in Refs.~\cite{CJL16,CJL16b}.

\begin{figure}
\centering
\includegraphics[width=0.5\textwidth]{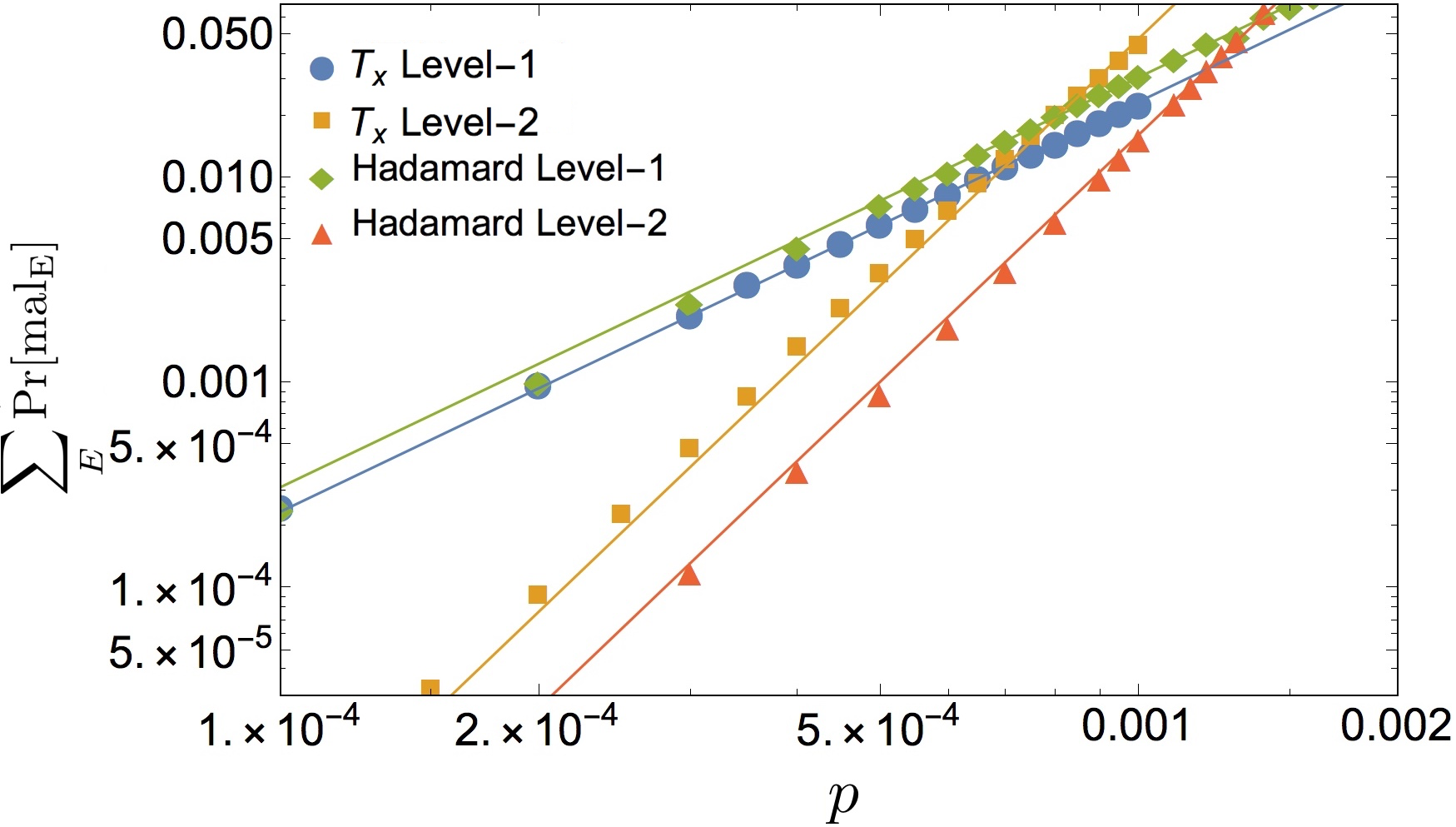}
\caption{Log-log plots comparing the total probability of failure for the Hadamard exRec of the [7,15] code with the $T_{x}$ exRec of the [15,15H] code at the first two levels. We chose the $H$ and $T_{x}$ gates since they provide lower bounds on the thresholds of the [7,15] and [15,15H] code.}
\label{fig:225and105Comparison}
\end{figure}

When using Steane error correction, the ancilla states $\ket{\overline{0}}$ and $\ket{\overline{+}}$ are encoded using the same error correcting code that protects the data. For codes using a large number of physical qubits to encode the data, the circuits encoding the ancilla states in Steane error correction will be very large. As an example, for the [15,15H] code which encodes 1 logical qubit to 225 physical qubits, the ancilla circuits contain approximately 2000 locations. For physical error rates $p \gtrsim 10^{-3}$, the ancilla states will get rejected with very high probability which leads to unrealistically large physical qubit overheads. To avoid large overheads arising from ancilla rejection, we will only compare the [7,15] and [15,15H] codes for physical error rates $p < 10 \times 10^{-3}$.

In \cref{fig:225and105Comparison}, we compare the \textit{total} logical failure rates of the $H$ and $T_x$ exRecs by summing $\mathrm{Pr}(\mathrm{mal}_{E})$ over all malignant events. We chose the Hadamard exRec for the [7,15] code since it lower bounds the codes threshold value \cite{CJL16}. For the $T$ gate exRec implemented in the [15,15H] code, we found that the errors were dominated by logical $Z$ errors whereas for the $T_x$ exRec, the errors were dominated by logical $X$ errors. 

\begin{figure}[h]
\centering
\begin{subfigure}{0.35\textwidth}
\includegraphics[width=\textwidth]{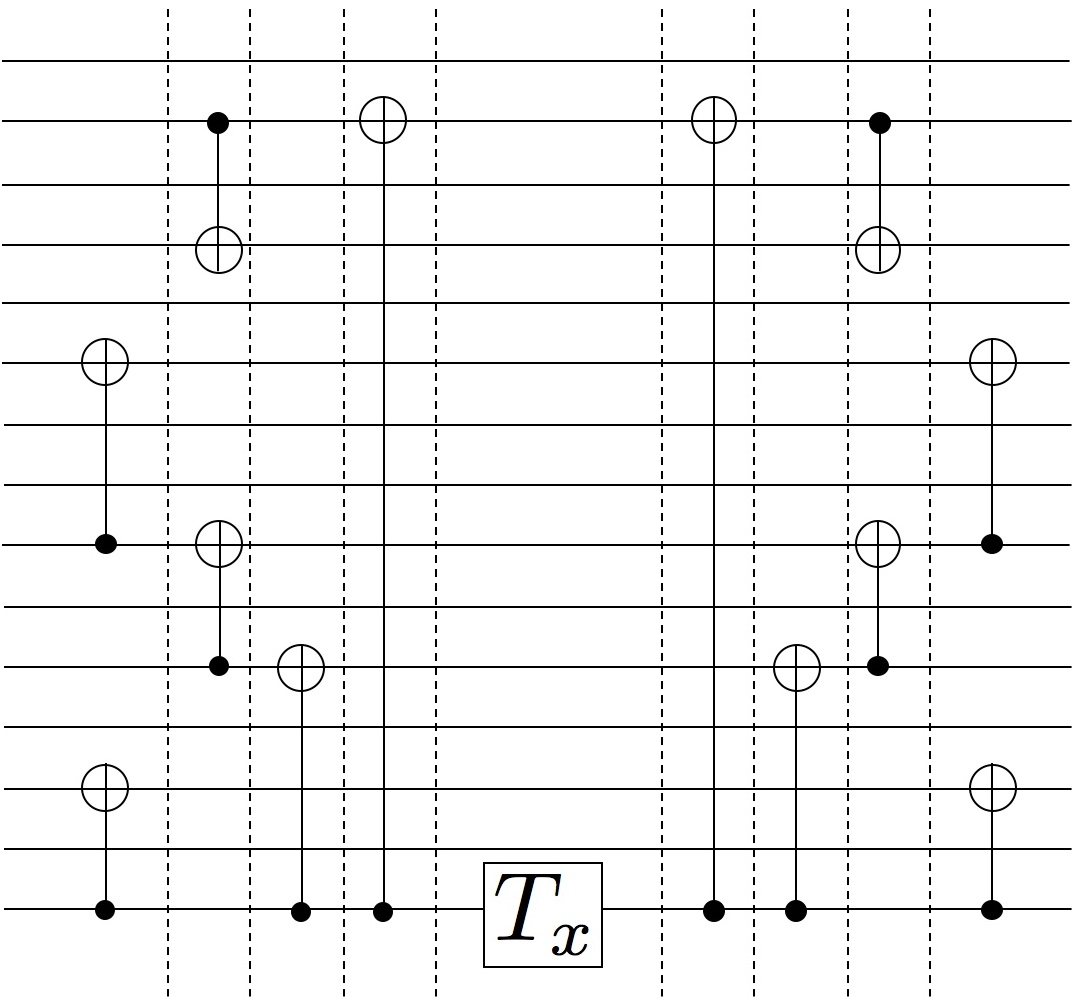}
\caption{}
\label{fig:HTHcircuit1}
\end{subfigure}
\begin{subfigure}{0.4\textwidth}
\includegraphics[width=\textwidth]{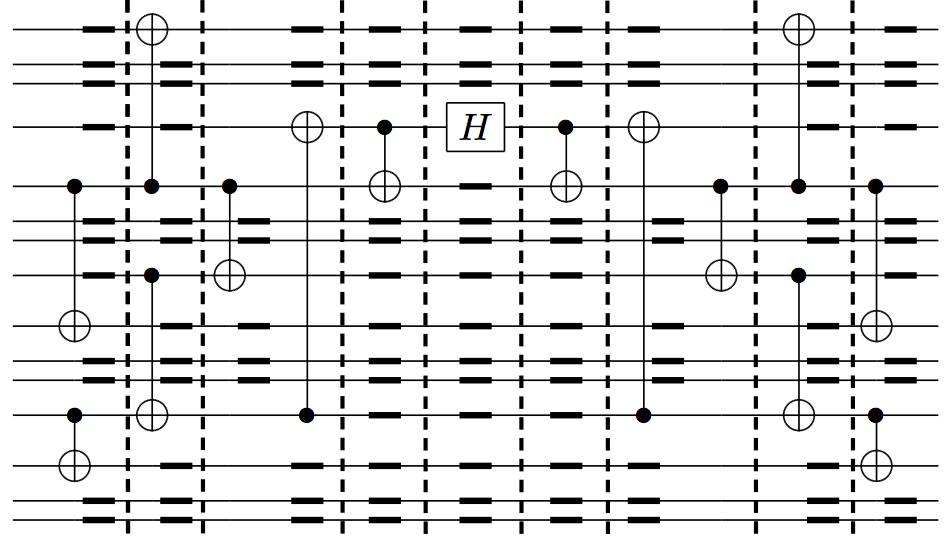}
\caption{}
\label{fig:HadCircuit}
\end{subfigure}
\caption{ (a) Circuit implementing a logical $T_{x}$ gate in the [15,15H] code. Each line is encoded in the Hadamard transform of the 15-qubit Reed-Muller code and $T_{x}$ is implemented transversally on the 15'th code block. (b) Circuit implementing a logical Hadamard gate in the 15-qubit Reed-Muller code. Note that to implement a logical Hadamard gate in the [7,15] code, each of the 7 code blocks implement the circuit in \cref{fig:HadCircuit} since the Hadamard gate can be implemented transversally in the Steane code. Solid dark lines represent resting qubit locations.}
\label{fig:HTHandHadcircuits}
\end{figure}

A detailed analysis of the circuits implementing a logical $T_{x}$ and $T$ gate in the 15-qubit code and its Hadamard transform is provided in \cref{app:MinCNOTforTandHTH} (see also \cref{fig:TandHTHcircuits}). Our numerical results show that the $T_x$ exRec has higher logical failure rates compared to the $T$ gate exRec. To understand these results, we first note that  from \cref{tab:BiasedNoiseTab}, the [15,15H] code is able to correct more $Z$ errors since the inner code has higher $d_{Z}$ distance. Consequently, after performing Steane error correction, the remaining errors at the output of an EC circuit will be biased towards $X$ errors. The $T_{x}$ exRec is very sensitive to input $X$ errors since there are more locations where they can propagate yielding a logical fault. Alternatively, the $T$ gate exRec is more sensitive to input $Z$ errors. Since at the first concatenation level input errors to the $T$ and $T_{x}$ exRecs will be dominated by $X$ errors, it is not surprising that the logical failure rate of the $T_{x}$ exRec will be larger.

When going to higher concatenation levels, there will be much fewer errors at the output of an EC circuit due to the double protection the CNOT gates from both codes. Therefore, input errors to the $T_{x}$ and $T$ gate exRecs will occur with much lower probability. However, since logical $X$ errors will occur with higher probability in a level-1 $T_{x}$ exRec compared to logical $Z$ errors in a level-1 $T$ gate exRec, at higher concatenation levels the $T_{x}$ exRec will continue to fail with higher probability compared to a $T$ gate exRec.

At the first concatenation level, it can be seen from \cref{fig:225and105Comparison} that the $T_{x}$ exRec slightly outperforms the Hadamard exRec of the [7,15] code for physical error rates $p > 10^{-4}$. We point out however that at low physical error rates, the pseudo-threshold of the [7,15] code ($(4.47 \pm 0.29) \times 10^{-5}$) is slightly higher than the pseudo-threshold of the [15,15H] code ($(3.91 \pm 0.24) \times 10^{-5}$), as obtained by our least-squares fit to the numerical data. We note that the [15,15H] code has roughly twice as many qubits as the outer code is larger, therefore increasing the number of possible error paths. However it also benefits from increased protection against $Z$~errors, therefore these two competing factors results in a slightly improved performance at the first concatenation level.

Due to the double protection from both the 15-qubit Reed-Muller code and its Hadamard transform code, compared to the $T$ and $T_{x}$ exRecs, all other exRecs of the [15,15H] code have lower logical failure rates by several orders of magnitude for the sampled physical error rates. Consequently, their contribution to the total logical failure rate of the $T$ and $T_{x}$ exRecs at the second concatenation level will be negligible. Now consider the circuit implementing the logical $T_{x}$ gate shown in \cref{fig:HTHcircuit1}, where each line is encoded in the Hadamard transform of the 15-qubit Reed-Muller code (more details are provided in \cref{app:MinCNOTforTandHTH}). A single $X$ error occurring at the $T_{x}$ gate location will result in a weight-one error on 7 code blocks after propagating through the CNOT gates. However, each error will be corrected on the individual code blocks due to the protection from the inner code. On the other hand, if a weight-2 $X$ error occurred at the $T_{x}$ location, the error would propagate to weight-two errors on 7 code blocks. Since the inner code has an $X$ distance $d_{X} = 3$, a logical fault would occur on each code block leading to an overall logical $X$ error. Therefore, at the second concatenation level, $\mathrm{Pr}[\mathrm{mal}^{(2)}_{X}|T_{x},p]$ can be approximated by

\begin{align}
\mathrm{Pr}[\mathrm{mal}^{(2)}_{X}|T_{x},p] \approx \sum_{l=2}^{15}\binom{15}{l} (\Gamma_{T_{x}}^{(1)})^{l},
\label{eq:HTHestimateLevel2}
\end{align}
where $\Gamma_{T_{x}}^{(1)}$ upper bounds $\mathrm{Pr}[\mathrm{mal}^{(1)}_{X}|T_{x},p]$. In \cref{eq:HTHestimateLevel2} we neglected the contributions from the other exRecs (CNOT, storage, measurement and state-preparation) since their level-1 logical failure rates are smaller by several orders of magnitude for the sampled physical error rates. Using $\Gamma_{T_{x}}^{(1)} \approx 23861p^{2}$ to leading order in $p$, obtained from our level-1 simulation, we find that \cref{eq:HTHestimateLevel2} reproduces the data of \cref{fig:225and105Comparison} obtained from our level-2 simulation. This result confirms that the performance of the [15,15H] code is mainly limited by $X$ errors of weight greater or equal to two occurring at the $T_{x}$ location in \cref{fig:HTHcircuit1}. Therefore, the primary benefit of the [15,15H] code is negated, since at the weak point of the circuit in~\cref{fig:HTHcircuit1} where the $T_x$ is implemented, there is only protection from the inner code which is susceptible to $X$~errors. The same will be true for the Hadamard gate in~\cref{fig:HadCircuit} for the [7,15]~code, however the number of locations in the inner code will be about half those in the [15,15H]~code, therefore the probability of error will be smaller. As such, the level-2 logical failure rate for the Hadamard circuit is smaller then that of the $T_x$ gate, as shown in~\cref{fig:225and105Comparison}.

Due to the large number of resting qubit locations in the $T_{x}$ exRec, we also performed a level-1 simulation of the $T_{x}$ exRec where the failure rate of resting qubit locations were 100 times less than gate locations. Even in this case, our numerical analysis showed that at the second concatenation level, the Hadamard exRec of the [7,15] code still achieved smaller logical failure rates compared to the $T_{x}$ exRec of the [15,15H] code for all sampled error rates.

\section{Conclusion}
\label{sec:Conclusion}

In this paper we began by considering error correcting codes obtained by concatenating two different codes $C_{1}$ and $C_{2}$. In \cref{subsec:OptimizeConcat} we provided general arguments showing that by choosing the inner code to be the code with the highest pseudo-threshold, the resulting code would have a higher code capacity threshold value. In \cref{subsec:PauliChanG} we briefly reviewed the process matrix formalism which we used in \cref{subsec:UniversalQuantumCodes} to compute exact threshold values of several universal concatenated quantum codes. All universal concatenated codes where the inner code had the highest pseudo-threshold achieved better error suppression capabilities. In particular, for universal concatenated codes where $C_1$ and $C_2$ were CSS codes, we found that concatenating the 15-qubit code with its Hadamard transform code (which we called [15,15H]) resulted in a code with the highest code capacity threshold value compared to the other codes in \cref{tab:DepolarizingThresholdResults}. The universal gate set of the [15,15H] code is not the standard Clifford + $T$ gate set but is instead generated by $\langle T_{x}, T, CNOT \rangle$ where $T_{x} = HTH$. This technique would generalize to higher distance 3D~color codes, among which the 15-qubit code is the smallest non-trivial code. Moreover, we believe that similar gauge switching techniques for the smaller stabilizers could be generalized to higher-order Reed-Muller codes, allowing for the implementation of fault-tolerant gates from higher levels of the Clifford hierarchy.

In \cref{sec:Circuit15Noise}, we performed a full circuit level noise analysis of the [15,15H] code and showed that its threshold was limited by the $T_{x}$ gate exRec. We compared the total logical failure rate of the $T_{x}$ exRec at the first and second concatenation level to that of the Hadamard exRec for the [7,15] code (which limits its threshold value). Although the code capacity threshold of the [15,15H] code is higher than the [7,15] code by more than a factor of 2, we showed that the [7,15] code outperforms the [15,15H] code when considering gate, preparation and measurement errors. We provided an analytic argument showing that the large number of error configurations of weight greater or equal to two occurring at the $T_{x}$ location of \cref{fig:HTHcircuit1} compared to error configurations of weight greater or equal to two occurring at different $H$ locations\footnote{Recall that each line of \cref{fig:HTHandHadcircuits} is composed of many qubits and transversal gates} of~\cref{fig:HadCircuit} was the primary reason that the [7,15] code achieved lower logical failure rates.

We conclude that for universal concatenated codes, it is preferable to use smaller codes in order to minimize the number of locations where errors can propagate in circuits used for fault-tolerantly implementing one of the underlying codes non-transversal gates. We note that a possibility for improved performance of the complementary concatenation scheme would be to use higher-distance 3D~color codes, strengthening potential weak points in the circuits with further protection. However, we note that again that there is a careful balance between growing the code size and improving the distance that would require further analysis.

\section{Acknowledgements}
The authors would like to thank Joel Wallman for useful discussions and Steve Weiss for providing the necessary computational resources. C. C. would like to acknowledge the support of QEII-GSST. T.~J. would like to acknowledge the support from the Walter~Burke Institute for Theoretical Physics in the form of the Sherman Fairchild Fellowship.

\bibliographystyle{ieeetr}
\bibliography{bibtex_chamberland}

\begin{thebibliography}{10}

\bibitem{JL14}
T.~Jochym-O'Connor and R.~Laflamme, ``Using concatenated quantum codes for
  universal fault-tolerant quantum gates,'' {\em Phys. Rev. Lett.}, vol.~112,
  p.~010505, 2014.

\bibitem{CJL16}
C.~Chamberland, T.~Jochym-O'Connor, and R.~Laflamme, ``Thresholds for universal
  concatenated quantum codes,'' {\em Phys. Rev. Lett.}, vol.~117, p.~010501,
  2016.

\bibitem{BM07}
H.~Bomb{\'\i}n and M.~A. Martin-Delgado, ``Topological quantum distillation,''
  {\em Phys. Rev. Lett.}, vol.~97, p.~180501, 2006.

\bibitem{BM07c}
H.~Bombin and M.~A. Martin-Delgado, ``Exact topological quantum order in $d=3$
  and beyond: Branyons and brane-net condensates,'' {\em Phys. Rev. B},
  vol.~75, p.~075103, Feb 2007.

\bibitem{CWBL16}
C.~Chamberland, J.~J. Wallman, S.~Beale, and R.~Laflamme, ``Hard decoding
  algorithm for optimizing thresholds under general markovian noise,'' {\em
  Phys. Rev. A}, vol.~95, p.~042332, 2017.

\bibitem{RDM02}
B.~Rahn, A.~C. Doherty, and H.~Mabuchi, ``Exact performance of concatenated
  quantum codes,'' {\em Phys. Rev. A.}, vol.~66, p.~032304, 2002.

\bibitem{LMPZ96}
R.~Laflamme, C.~Miquel, J.~P. Paz, and W.~H. Zurek, ``Perfect quantum error
  correcting code,'' {\em Phys. Rev. Lett.}, vol.~77, no.~198, 1996.

\bibitem{Steane96b}
A.~W. Steane, ``{Multiple-Particle Interference and Quantum Error
  Correction},'' {\em Proc. Roy. Soc. Lond.}, vol.~452, pp.~2551--2577, 1996.

\bibitem{ADP14}
J.~T. Anderson, G.~Duclos-Cianci, and D.~Poulin, ``Fault-tolerant conversion
  between the steane and reed-muller quantum codes,'' {\em Phys. Rev. Lett.},
  vol.~113, p.~080501, 2014.

\bibitem{Bombin15}
H.~Bomb{\'\i}n, ``Gauge color codes: optimal transversal gates and gauge fixing
  in topological stabilizer codes,'' {\em New J. Phys.}, vol.~17, no.~8,
  p.~083002, 2015.

\bibitem{NSZ16}
E.~Nikahd, M.~Sedighi, and M.~S. Zamani, ``Non-uniform code concatenation for
  universal fault-tolerant quantum computing,'' {\em arXiv preprint
  arXiv:1605.07007}, 2016.

\bibitem{PR12}
A.~Paetznick and B.~W. Reichardt, ``Fault-tolerant ancilla preparation and
  noise threshold lower bounds for the 23-qubit golay code,'' {\em Quant. Inf.
  Compt.}, vol.~12, pp.~1034--1080, 2011.

\bibitem{CJL16b}
C.~Chamberland, T.~Jochym-O'Connor, and R.~Laflamme, ``Overhead analysis of
  universal concatenated quantum codes,'' {\em Phys. Rev. A}, vol.~95,
  p.~022313, 2017.

\bibitem{Steane02}
A.~M. Steane, ``Fast fault-tolerant filtering of quantum codewords,'' {\em
  arXiv preprint quant-ph/0202036}, 2002.

\bibitem{AGP06}
P.~Aliferis, D.~Gottesman, and J.~Preskill, ``Quantum accuracy threshold for
  concatenated distance--3 codes,'' {\em Quant. Inf. Comput.}, vol.~6,
  pp.~97--165, 2006.

\bibitem{BMPRV99}
P.~O. Boykin, T.~Mor, M.~Pulver, V.~Roychowdhury, and F.~Vatan, ``On universal
  and fault-tolerant quantum computing: A novel basis and a new constructive
  proof of universality for shor's basis,'' in {\em Foundations of Computer
  Science, 1999. 40th Annual Symposium on}, pp.~486--494, IEEE, 1999.

\bibitem{SteaneCSS}
A.~W. Steane, ``Enlagement of calderbank-shor-steane quantum codes,'' {\em
  IEEE. Trans.Inform. Theory}, vol.~45, no.~7, pp.~2492--2495, 1999.

\bibitem{CDT09}
A.~Cross, D.~P. DiVincenzo, and B.~M. Terhal, ``A comparative code study for
  quantum fault tolerance,'' {\em Quant. Inf. Comput.}, vol.~9, pp.~541--572,
  2009.

\end{thebibliography}

\clearpage
\appendix

\section{Proof of transversal $HTH$ for the rotated Reed-Muller code}
\label{app:TransversalHTH}

Let the rotated Reed-Muller code be the quantum error correcting code where all of the $X$~stabilizers of the original Reed-Muller code are replaced by $Z$~stabilizers, and vice versa. Additionally, define logical~$X$ as $X_{L,15H} = X^{\otimes 15}$, and logical~$Z$ as $Z_{L, 15H} = Z^{\otimes 15}$. That is, $\mathcal{S}_{15H} := \{ H^{\otimes 15} S H^{\otimes 15} | S \in \mathcal{S}_{15} \}$, where $\mathcal{S}_{15}$ and $\mathcal{S}_{15H}$ are the stabilizer groups of the 15-qubit and rotated 15-qubit Reed-Muller codes, respectively.
\begin{claim}
$\otimes_i (HTH)_i$ is a transversal gate for the rotated Reed-Muller code and implements a logical~$(HT^{\dagger}H$)~gate.
\end{claim}

\begin{proof}
We begin by showing that the stabilizers of the 15H~code are preserved under the action of~$\otimes_i (HTH)_i$. Let $Q \in \mathcal{S}_{15H}$, then,
\begin{align*}
(HTH)^{\otimes 15} Q (HT^{\dagger}H)^{\otimes 15} &= (HT)^{\otimes 15} S_Q (T^{\dagger}H)^{\otimes 15} \\
&= (H)^{\otimes 15} S'_Q (H)^{\otimes 15}\\
&= Q',
\end{align*}
Note $S_Q = H^{\otimes 15} Q H^{\otimes 15} \in \mathcal{S}_{15}$, by definition. Then, since $S_Q \in \mathcal{S}_{15}$ and the $T$~gate is transversal (implementing logical $T^{\dagger}$) for the 15-qubit Reed Muller code, it preserves the stabilizer group and $S'_Q = T^{\otimes 15} S_Q (T^{\dagger})^{\otimes 15} \in \mathcal{S}_{15}$. Finally, $ Q' = H^{\otimes 15} S'_Q H^{\otimes 15} \in \mathcal{S}_{15H}$, by definition, and as such the stabilizer group is preserved.

We now have to show that the logical operators transform in the appropriate manner to determine the logical transformation. Recall again that transversal~$T$ applies logical~$T^{\dagger}$ for the 15-qubit code. Therefore,
\begin{align*}
(HTH)^{\otimes 15} X^{\otimes 15} (HT^{\dagger} H)^{\otimes 15} &= (HT)^{\otimes 15} Z^{\otimes 15} (T^{\dagger} H)^{\otimes 15} \\
&= H^{\otimes 15} Z^{\otimes 15} H^{\otimes 15} \\
&= X^{\otimes 15}, 
\end{align*}
and the logical~$X$ operator is preserved. As for the logical~$Z$ operator,
\begin{align*}
(HTH)^{\otimes 15} Z^{\otimes 15}  (H & T^{\dagger} H)^{\otimes 15} \\
&= (HT)^{\otimes 15} X^{\otimes 15} (T^{\dagger} H)^{\otimes 15} \\
&= H^{\otimes 15} \left(\dfrac{1}{\sqrt{2}}(X^{\otimes 15} - Y^{\otimes 15}) \right)S H^{\otimes 15} \\
&= \left(\dfrac{1}{\sqrt{2}}(Z^{\otimes 15} - Y^{\otimes 15}) \right)Q_S ,
\end{align*}
where we have used the fact that transversal~$T$ implements a logical~$T^{\dagger}$ in the 15-qubit Reed-Muller code, therefore $X_L$ transforms to $\dfrac{1}{\sqrt{2}}(X^{\otimes 15} - Y^{\otimes 15})$, up to a stabilizer~$S$. Therefore, transversal~$(HTH)$ implements logical~$(HT^{\dagger}H)$ for the rotated 15-qubit code.
\end{proof}

\section{Generating a universal set of gates}
\label{app:CliffordGenerating}

The gate set $\langle H, T, \text{CNOT} \rangle$ is a universal gate set for quantum computation~\cite{BMPRV99}. Therefore, we show~$\langle T_x, T, \text{CNOT} \rangle$ is universal by showing it can generate~$H$.

\begin{claim}
$\langle T_x, T, \text{CNOT} \rangle$ form a universal gate set.
\end{claim}

\begin{proof}
Note that $T_x = HTH$ can easily generate $HSH = (HTH)^2$, since $H^2 = I$ and $T^2 = S$. We will now show that using the Clifford gates $HSH$ and $S$ we can generate~$H$. First note how the $H$~transforms the Pauli operators under conjugation:
\begin{align*}
H : \ & X \rightarrow Z \\
& Y \rightarrow -Y \\
& Z \rightarrow X.
\end{align*}
Now, note how $HSH$ and $S$ transform the Pauli operators under conjugation:
\begin{align*}
HSH : & \ X \rightarrow X \\
& Y \rightarrow Z  \\
& Z \rightarrow -Y , \\
\\
S : \  &  X \rightarrow Y \\
& Y \rightarrow -X \\
& Z \rightarrow Z. 
\end{align*}
Given these transformation properties, it is straightforward to note that $(HSH)S(HSH)$ will be equivalent (up to global phase) to a Hadamard transformation:
\begin{align*}
&X \xrightarrow{HSH} X \xrightarrow{S} Y \xrightarrow{HSH} Z \\
&Y \xrightarrow{HSH} Z \xrightarrow{S} Z \xrightarrow{HSH} -Y \\
&Z \xrightarrow{HSH} -Y \xrightarrow{S} X \xrightarrow{HSH} X.
\end{align*}
\end{proof}

\section{General method for finding encoded $|+\rangle$ circuits}
\label{app:EncodedCircuits}

Calderbank-Shor-Steane (CSS) codes are constructed from two classical codes $C_{1}$ and $C_{2}$ with the property that $C_{1}^{\perp} \subseteq C_{2}$ \cite{SteaneCSS}. The codes parity check matrix is given by 

\begin{align}
H = \left( \begin{array}{cc}
                                          H_{z} &  \\
                                           &H_{x}\\                                                                     
                                          \end{array} \right).
\label{eq:ParityCheck}
\end{align}
The $X$-stabilizers generators are given by the rows of $H_{x}$ with one's replaced by $X$ operators and zeros mapped to the identity. Similarly the $Z$-stabilizers generators are given by the rows of $H_{z}$. The condition  $C_{1}^{\perp} \subseteq C_{2}$ ensures that $H_{x} \cdot H_{z}^{T} = H_{z} \cdot H_{x}^{T} = 0$ so that all stabilizer generators commute. 

We first describe a method for finding the encoding circuits of $|\overline{0} \rangle$ for CSS codes. Suppose that the code has $r$ $X$-stabilzer generators. The standard method has been to perform Gaussian elimination on the matrix $H_{x}$ to reduce it to the form

\begin{align}
H_{x} \to \left( \begin{array}{cc}
                                          I&| \hspace{0.2cm} A                                                                                           
                                          \end{array} \right),
\label{eq:FirstGaussianMethod}
\end{align} 

where $I$ is the $r \times r$ identity matrix \cite{CDT09, PR12, Steane02}. For each row, the ones of the identity matrix act as the control qubits of the CNOT, which are prepared in the state $| + \rangle$. The remaining ones are the target qubits prepared in the state $|0 \rangle$. However, in most cases, the Gaussian elimination steps that are performed to reduce $H_{x}$ to the form in \cref{eq:FirstGaussianMethod} require permutation of some of the columns of $H_{x}$. 

We point out that permutation of the columns of $H_{x}$ in the Gaussian elimination procedure is unnecessary. Instead, one can perform Gaussian elimination on the matrix $H_{x}$ without swapping any columns. For each row of the resulting matrix, the first non-zero element acts as the control qubit of the CNOT prepared in the $| + \rangle$ state. The remaining ones of the row under consideration acts as the target qubits of the CNOT. Our method requires fewer operations and avoids having to swap qubits in order to obtain the correct encoded state.

Self-dual CSS codes satisfy the property that $C_{1}^{\perp} = C_{2}$. For self-dual codes, the encoding circuit for the encoded $|\overline{+} \rangle$ can be obtained from the encoding circuit of $|\overline{0} \rangle$ by interchanging all the physical $|+\rangle$ states with the $|0\rangle$ states and reversing the direction of all the CNOT gates. We now present a method for finding the encoding circuit of the $|\overline{+} \rangle$ state which is applicable to any CSS codes. 

Once the encoding circuit for the $\ket{\overline{0}}$ state is obtained, it can be written as 

\begin{align}
\ket{\overline{0}} = U_{CNOT}\ket{\psi_{in}},
\label{eq:PlusEncodingStep1}
\end{align}
where $U_{CNOT}$ is the unitary matrix containing all the CNOT gates in the encoding circuit of $\ket{\overline{0}}$ and $\ket{\psi_{in}}$ is its input state, consisting of a tensor product of $\ket{0}$ and $\ket{+}$ states. 

The $\ket{\overline{1}}$ can be obtained from the $\ket{\overline{0}}$ state by using the codes logical $\overline{X}$ as follows
\begin{align}
\ket{\overline{1}} &= \overline{X}\ket{\overline{0}} \nonumber \\
 &=  \overline{X}U_{CNOT}\ket{\psi_{in}} \nonumber \\
 &= U_{CNOT} \tilde{X}\ket{\psi_{in}},
\label{eq:PlusEncodingStep2}
\end{align}
where the operator $\tilde{X}$ is obtained by commuting the codes logical $\overline{X}$ through the CNOT gates backwards in time. In other words, we propagate $\overline{X}$ from the final CNOT gates in the $\ket{\overline{0}}$ circuit to the direction of the input state $\ket{\psi_{in}}$. 

Since $\ket{\overline{+}} = \ket{\overline{0}} + \ket{\overline{1}}$, combining \cref{eq:PlusEncodingStep1,eq:PlusEncodingStep2} we obtain 

\begin{align}
\ket{\overline{+}} = U_{CNOT}(\ket{\psi_{in}} + \tilde{X}\ket{\psi_{in}}).
\label{eq:PlusFinal}
\end{align}
This shows that the encoding circuit for $\ket{\overline{+}}$ can be obtained from the same encoding circuit as $\ket{\overline{0}}$ but with the modified input state $\ket{\psi_{in}} + \tilde{X}\ket{\psi_{in}}$. The circuit depth can then be optimized by applying Steane's Latin rectangle method. We note that our method is in general not optimal in terms of circuit depth and CNOT gate counts. However, methods presented in \cite{CJL16b,PR12} can be used to reduce the number of CNOT gates in the encoding circuits. 

\section{Minimum number of CNOT gates for $T_{x}$ and $T$ circuits}
\label{app:MinCNOTforTandHTH}

\begin{figure}[h]
\centering
\begin{subfigure}{0.35\textwidth}
\includegraphics[width=\textwidth]{HTHcircuit225.png}
\caption{}
\label{fig:HTHcircuit225}
\end{subfigure}
\begin{subfigure}{0.4\textwidth}
\includegraphics[width=\textwidth]{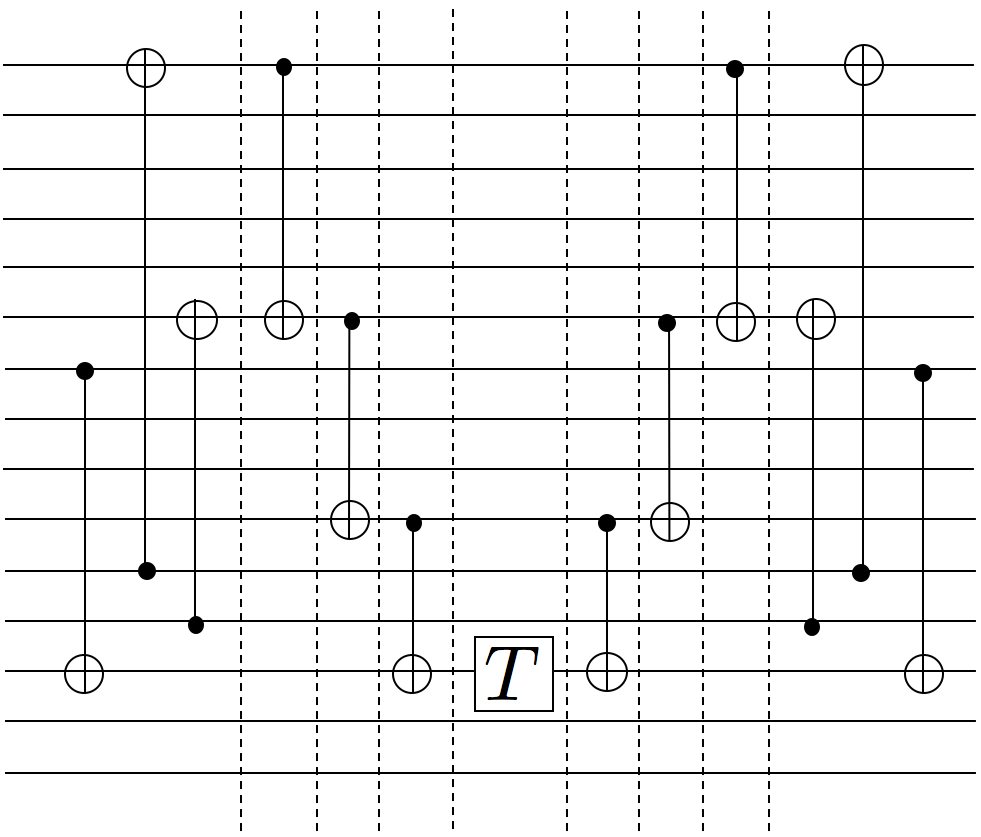}
\caption{}
\label{fig:Tcircuit225}
\end{subfigure}
\caption{ (a) $T_{x}$ and (b) $T$ gate circuits used in our simulations for the 225-qubit code. In both circuits the CNOT gates on the left of $T$ or $T_{x}$ are the inverse of the CNOT gates on its right. Each circuit has a total of 12 CNOT gates and is implemented in 9 time steps. If a single physical $T$ and $T_{x}$ gate is used, these circuits cannot be implemented with fewer than 12 CNOT gates.}
\label{fig:TandHTHcircuits}
\end{figure}

In this section we provide a proof that the number of CNOT gates used in the circuits implementing the $T_{x}$ gate in the 15-qubit Reed-Muller code and the circuits implementing the $T$ gate in the Hadamard transform of the 15-qubit Reed-Muller code cannot use fewer than six CNOT gates.

Since $T=\mathrm{diag}(1,e^{i\pi/4})$, it can be shown that 

\begin{align}
TXT^{\dagger} &= \frac{1}{\sqrt{2}}X(I+iZ), \label{eq:TtransformationRulesX} \\
TZT^{\dagger} &= Z, \label{eq:TtransformationRulesZ} \\
T_{x}XT_{x}^{\dagger} &= X, \label{eq:T_{x}transformationRulesX} \\
T_{x}ZT_{x}^{\dagger} &= \frac{1}{\sqrt{2}}Z(I+iX). \label{eq:T_{x}transformationRulesZ}
\end{align}

A general feature about the circuits in \cref{fig:TandHTHcircuits} is that they have the form $U_{CNOT}T_{x_{i}}U_{CNOT}^{\dagger}$ and $U_{CNOT}T_{j}U_{CNOT}^{\dagger}$. Given the transformation rules of \cref{eq:TtransformationRulesX,eq:TtransformationRulesZ,eq:T_{x}transformationRulesX,eq:T_{x}transformationRulesZ}, the form of the circuits ensures that for the $T_{x}$ gate, the 15-qubit codes $Z$ stabilizers have no support on i'th qubit when propgating through $U_{CNOT}$. Similarly, the $X$ stabilizers of the Hadamard transform of the 15-qubit code will have no support on the j'th qubit when propagating through $U_{CNOT}$. The gates $U_{CNOT}^{\dagger}$ will then undo the transformations underwent from $U_{CNOT}$ so that all stabilizer generators are unchanged when propagating through the $T_{x}$ and $T$ circuits. 

We first focus on the implementation of $T_{x}$ in the 15-qubit Reed-Muller code. Note that logical $X$ operators for the 15-qubit code have minimum weight 7 \cite{ADP14} (one possible representation is $\overline{X} = X_{1}X_{2}X_{3}X_{4}X_{5}X_{6}X_{7}$). From \cref{eq:T_{x}transformationRulesX}, $T_{x}$ leaves $X$ operators invariant under conjugation so that the circuit in \cref{fig:HTHcircuit225} will have no effect on $\overline{X}$. 

The logical $Z$ operator must transform as $\overline{Z} \to \frac{1}{\sqrt{2}}\overline{Z}(I-i\overline{X})$. The gates $U_{CNOT}$ in \cref{fig:HTHcircuit225} must be chosen to ensure that $\overline{Z}$ always has support on the i'th qubit (the qubit acted on by $T_{x}$) when propagating through $U_{CNOT}$. From \cref{eq:T_{x}transformationRulesZ} $Z_{i} \to \frac{1}{\sqrt{2}}Z_{i}(I-iX_{i})$. Hence we will have

\begin{align}
 &\overline{Z}U_{CNOT}T_{x_{i}}U_{CNOT}^{\dagger} \nonumber \\
 &= U_{CNOT}\tilde{Z}Z_{i}T_{x_{i}}U_{CNOT}^{\dagger}  \nonumber \\
 &= \frac{1}{\sqrt{2}}U_{CNOT}T_{x_{i}}\tilde{Z}Z_{i}(I-iX_{i})U_{CNOT}^{\dagger}
 \label{eq:T_{x}transformZbar}
 \end{align}

Since $\overline{Z}U_{CNOT} = U_{CNOT} \tilde{Z}Z_{i}$, then $\tilde{Z}Z_{i}U_{CNOT}^{\dagger} = U_{CNOT}^{\dagger}\overline{Z}$. However, we now require that $X_{i}U_{CNOT}^{\dagger} = U_{CNOT}^{\dagger}\overline{X}$ to ensure that the logical $Z$ operator transforms as intended by a logical $T_{x}$ gate. Since the minimum weight of $\overline{X}$ is 7 for the 15-qubit code, there must be a \textit{minimum of 6} CNOT gates in $U_{CNOT}^{\dagger}$ to ensure that the single-qubit $X_{i}$ operator transforms to a weight 7 $X$ operator. Therefore, the minimum number of CNOT gates in the $T_{x}$ circuit is 12. An analogous argument using the transformation rules in \cref{eq:TtransformationRulesX,eq:TtransformationRulesZ} can be used to show that the minimum number of CNOT gates for the $T$ circuit is also 12.

\section{Proof of minimal depth CNOT construction}
\label{app:MinCNOTDepth}

In this section we show that the transformations given in Figs.~\ref{fig:HTHcircuit225}--\ref{fig:Tcircuit225} are not only optimal in terms of their minimal number of CNOT gates, but also in their depth of the $U_{CNOT}$~transformation. As noted in \cref{app:MinCNOTforTandHTH}, the transformation $U_{CNOT}$ must satisfy:
\begin{align}
U_{CNOT} \overline{X} U_{CNOT}^{\dagger} &= X_i \\
U_{CNOT} \overline{Z} U_{CNOT}^{\dagger} &= Z_i,
\end{align}
for some chosen target qubit~$i$. We will prove the result for the logical gate~$T_x$ as the argument will follow similarly for the logical gate~$T$. As such, we have the following equalities to satisfy:
\begin{align}
U_{CNOT}^{\dagger}  X_{15} U_{CNOT} &= \overline{X} = X_{13} X_{14} X_{15}\\
U_{CNOT}^{\dagger} Z_{15} U_{CNOT} &= \overline{Z} = Z_{9} Z_{10} Z_{11} Z_{12} Z_{13} Z_{14} Z_{15},
\end{align}
where we have chosen particular minimal representations of the logical operators, however any such choice would be equivalent. We have chosen to study $U_{CNOT}^{\dagger}$ as we find it more intuitive to think of the inverse transformation, yet again all arguments will translate.

We already have a depth~$4$ circuit implementing the desired logical transformation and our goal will be to find a operation $U_{CNOT} = U_3 U_2 U_1$, where each $U_k$ is composed of only CNOT gates, with no overlapping support on the CNOT gate supports (that is each segment is trivially of depth~1). We do not consider the case of depth~2 as this can be easily ruled out by the following identity: $\text{wt}((U_2U_1)^{\dagger}P_i(U_2U_1)) \le 4$, where $\text{wt}(\cdot)$ is the weight of a Pauli operator. This equality follows from the fact that a given depth~1 circuit of CNOT gates can only at most double the support of a particular Pauli operator. Therefore, given a depth~2 circuit, a weight-1 $Z$~operator could at most grow to be of weight~4, smaller than the desired target weight of the distance~7. As such, this rules out circuits of depth~2.

However, circuits of depth~3 are not ruled out by the above argument. Yet, there is some structure that a depth~3 circuit would need to have in order to achieve the transformation: $U_{CNOT}^{\dagger} Z_{15} U_{CNOT} =  U_1^{\dagger} U_2^{\dagger}U_3^{\dagger} Z_{15} U_3 U_2 U_1 = \overline{Z}$. Note that $U_3$ must have a CNOT gate with qubit~15 as a target qubit in order to map $Z_{15}$ to a Pauli operator composed of $Z$~operators on 2~qubits. If not then $\text{wt}(U_1^{\dagger} U_2^{\dagger}(U_3^{\dagger} Z_{15} U_3) U_2 U_1) = \text{wt}(U_1^{\dagger} U_2^{\dagger} Z_{15} U_2 U_1) \le 4 < 7$, and we arrive at the same contradiction as before. Therefore, $U_3^{\dagger} Z_{15} U_3 = Z_l Z_{15}$, and as a result $U_3^{\dagger} X_{15} U_3 = X_{15}$, since qubit~15 is a target qubit, the $X_{15}$~remains unchanged. The operator $U_2$ must also have a set of CNOT~gates with target qubits on qubits~$l$~and~15, else we arrive at the following contradiction: $\text{wt}(U_2^{\dagger}(Z_l Z_{15})U_2) \le 3 \Rightarrow \text{wt}(U_1^{\dagger} U_2^{\dagger}(Z_l Z_{15})U_2 U_1) \le 2\cdot 3 = 6 < 7$. As such, we must have that~$U_2^{\dagger}(Z_l Z_{15})U_2 = Z_l Z_m Z_n Z_{15}$ and similarly $U_2^{\dagger} X_{15} U_2 = X_{15}$. However, we now arrive at our final contradiction, $\text{wt}(U_1^{\dagger}X_{15}U_1) \le 2 < 3$. Therefore, there exists no depth~3 $U_{CNOT}$ circuit that can achieve the desired transformation of the logical Pauli operators.

This is important from the perspective of error correction as minimizing the depth minimizes the number of memory locations in the circuit which have a large effect on the logical error rate at the first few concatenation levels.
\end{document}